%% file: main.tex
\documentclass{article}
\usepackage[utf8]{inputenc}

\usepackage[a4paper]{geometry}
\usepackage{setspace} 
\usepackage[round]{natbib}
\usepackage{hyperref}
\hypersetup{colorlinks=true, citecolor=blue}

\usepackage{graphicx}
\usepackage{subfig}

\usepackage{mathtools} 
\usepackage{amssymb} 
\usepackage{amsthm} 
\usepackage{bm} 

\newtheorem{theorem}{Theorem}
\newtheorem{lemma}{Lemma}

\let\Pr\relax
\DeclareMathOperator{\Pr}{\mathbb{P}}
\DeclareMathOperator{\E}{\mathbb{E}}
\newcommand{\R}{\mathbb{R}}

\usepackage[affil-it]{authblk} 
\usepackage[symbol]{footmisc} 

\title{Improving the output quality of official statistics based on machine learning algorithms}

\author[1,2,3]{Q.A. Meertens\thanks{Corresponding author: \texttt{q.a.meertens@uva.nl}}}
\affil[1]{Statistics Netherlands}
\affil[2]{University of Amsterdam}
\affil[3]{Leiden University}

\author[2]{C.G.H. Diks}

\author[3]{H.J. van den Herik}

\author[3]{F.W. Takes}

\date{14 December 2020}

\begin{document}

\maketitle

\section*{Abstract}
National statistical institutes currently investigate how to improve the output quality of official statistics based on machine learning algorithms. A key obstacle is concept drift, i.e., when the joint distribution of independent variables and a dependent (categorical) variable changes over time. Under concept drift, a statistical model requires regular updating to prevent it from becoming biased. However, updating a model asks for additional data, which are not always available. In the literature, we find a variety of bias correction methods as a promising solution. In the paper, we will compare two popular correction methods: the misclassification estimator and the calibration estimator. For prior probability shift (a specific type of concept drift), we investigate the two correction methods theoretically as well as experimentally. Our theoretical results are expressions for the bias and variance of both methods. As experimental result, we present a decision boundary (as a function of (a) model accuracy, (b) class distribution and (c) test set size) for the relative performance of the two methods. Close inspection of the results will provide a deep insight into the effect of prior probability shift on output quality, leading to practical recommendations on the use of machine learning algorithms in official statistics.
\footnote[3]{The views expressed in this paper are those of the authors and do not necessarily reflect the policy of Statistics Netherlands. The authors would like to thank Sander Scholtus for his useful comments on an earlier version of the manuscript.}

\vspace{10pt}
\noindent \textbf{Keywords:} \textit{Machine learning, output quality, concept drift, prior probability shift, misclassification bias}

\newpage

%
%
\section{Introduction}\label{sec:introduction}
In recent years, many national statistical institutes (NSIs) have experimented with supervised machine learning algorithms with the purpose of producing new or improved official statistics. \citet{beck_machine_2018} provide a list of 136 machine learning projects at NSIs in 25 countries. In many projects, machine learning was used for classification (78) or for imputation (22). The results of these machine learning projects are promising and therefore currently seen as a paradigm shift in official statistics, in which model-based statistics are widely embraced \citep{debroe_updating_2020}.

The quality of the statistical output is a key challenge when employing classification algorithms for producing official statistics. Output quality is a fundamental component in any quality framework for official statistics, see, e.g., the OECD quality framework \citep{oecd_quality_2011} and the Regulation on European Statistics \citep{european_commission_regulation_2009} translated into the European Statistics Code of Practice \citep{eurostat_european_2017}. When using classification algorithms for official statistics, the output quality ought to be measured using the mean squared error of the statistical output \citep{buelens_model_2016}.

In the machine learning literature, the accuracy of classification algorithms is measured at the level of individual data points. Interestingly, the algorithmic accuracy at the level of \textit{individual} data points differs fundamentally from the accuracy (at the \textit{population} level) of the (aggregated) statistical output of classification algorithms \citep{forman_counting_2005}. In fact, classification algorithms that have high algorithmic accuracy might still produce highly biased statistical output. This is referred to as \textit{misclassification bias}. It is a type of bias that is commonly overlooked or neglected by statisticians of all time  \citep{schwartz_neglected_1985,gonzalez_review_2017}.

After many years of persistent research, a rich body of statistical literature on misclassification bias is readily available. Misclassification bias occurs in general when dealing with measurement errors in categorical data. The work by \citet{bross_misclassification_1954} is usually referred to as the first publication to discuss the problem of misclassification bias. Other significant contributions to the literature on misclassification bias include the work by \citet{tenenbein_double_1970} and the work by \citet{kuha_categorical_1997}. A relatively recent overview is provided by \citet{buonaccorsi_measurement_2010}.

The literature on misclassification bias shows that the bias can be reduced significantly, if some form of extra information is available. In the general context of categorical data analysis, this extra information can be, for instance, replicate values, validation data, or instrumental variables \citep{buonaccorsi_measurement_2010}. Although such extra information in general might not always be available, it is available in the context of supervised machine learning that we are considering here.
The extra information are validation data, which are traditionally used for model selection, training and testing. We will use the test set as validation data to estimate error rates, and thus to correct misclassification bias.

In experimental projects at NSIs, the test set often is a random sample from the target population (e.g., all households in the country). The setup corresponds to the double sampling scheme introduced by \citet{tenenbein_double_1970}. Among the correction methods discussed by \citet{buonaccorsi_measurement_2010}, the so-called \textit{calibration estimator} then outperforms all the others in terms of mean squared error, as proved theoretically by \citet{kloos_comparing_2020}.

However, a new problem arises when incorporating machine learning algorithms in the production process of official statistics. There, a statistical model is often estimated once and then applied for a longer period of time without updating the model parameters. In the context of supervised machine learning this is common, because otherwise new data have to be annotated manually in each time period leading to high production costs. However, the problem there is that the data distribution as well as the relation between the dependent and independent variables might change over time, causing the outcome of the model to be biased. In the machine learning literature, this problem is known as \textit{concept drift}. It has been investigated in stream learning and online learning for several decades \citep[see][]{widmer_learning_1996}, dating back at least to the work on incremental learning \cite[cf.][]{schlimmer_incremental_1986} in the 1980s. Originally, the term \textit{concept} was used for a set of Boolean-valued functions \citep{helmbold_tracking_1994}. Currently, it has a statistical interpretation that is more closely related to our setting. Nowadays, \citet{webb_characterizing_2016} state that the term \textit{concept} refers to the joint distribution $\Pr(Y, X)$, with class labels (dependent variable) $Y$ and features (independent variables) $X$, as proposed by \citet{gama_survey_2014}. Allowing such a joint distribution to depend on a time parameter $t$, concept drift in the setting of supervised learning means that $\Pr_{t_1}(Y, X) \neq \Pr_{t_2}(Y, X)$, for $t_1 \neq t_2$. The effect of concept drift is that misclassification bias might increase even further.

In this paper, we aim to prove which of the two popular correction methods discussed by \citet{buonaccorsi_measurement_2010} reduces the mean squared error of statistical output most, under a specific type of concept drift known as \textit{prior probability shift} \citep{moreno-torres_unifying_2012}. Our paper deliberately focuses on the production process (where concept drift arises), building on the results obtained by \citet{kloos_comparing_2020} for the preceding experimental phase. Our numerical analyses will show, for the first time, that a decision boundary arises. The optimal choice for a correction method depends on three parameters, viz. the class distribution (or class imbalance), the size of the test set, and the model accuracy. With that knowledge we aim to contribute to the literature on concept drift \textit{understanding} as defined by \citet{lu_review_2019}. It complements concept drift \textit{quantification} \citep{goldenberg_survey_2019} and concept drift \textit{adaptation} \citep{gama_survey_2014}. Analysing the decision boundary as a function of the three parameters yields practical recommendations for the implementation of classification algorithms in the production process of official statistics. Finally, analysing the impact of the \textit{size} of the (manually created) test set allows us to comment on the cost efficiency of official statistics based on classification algorithms.

The remainder of the paper is organised as follows. In Section~\ref{sec:methods}, we provide expressions for the bias and variance of the misclassification and calibration estimator, when applied to machine learning algorithms that have been implemented in the production process of official statistics. We show (1) that the optimal correction method in the experimental phase is no longer unbiased when implemented in a production process and we provide (2) a sharp lower bound for the absolute value of its bias. Hence, instead of arriving at a conclusive optimal solution in the experimental phase, a decision boundary arises in the context of the production process. Subsequently, in Section~\ref{sec:results}, we investigate the location and shape of that decision boundary. In Section~\ref{sec:discussion} we present our conclusions and suggest three promising directions for future research.

%
%
\section{Methods}\label{sec:methods}

In the context of official statistics, the convention is to use the mean squared error to evaluate output quality, also when using statistical models \citep{buelens_model_2016}. The key question when correcting misclassification bias then becomes: which correction method reduces the mean squared error of the output most? The outcome depends on the assumptions made. The situation that fits the experimental phase of machine learning projects at NSIs is discussed briefly in Subsection~\ref{sec:methods-exp_phase}. The assumptions made in the experimental phase are considered to be the most restrictive ones. The answer to the key question under those restrictive assumptions has been provided by \citet{kloos_comparing_2020} and it is rather conclusive. A drawback of their result is that data are assumed to be annotated manually in each time period. In practice, manual data annotation is time consuming and hence expensive. Therefore, in Subsection~\ref{sec:methods-realistic}, we describe the situation that corresponds to the production process of official statistics. In Subsection~\ref{sec:methods-theory}, the theoretical results known for the experimental phase are adapted to suit the conditions of the production process of official statistics. The answer to the key question in that setting is presented in Section~\ref{sec:results}.

\subsection{The experimental phase}\label{sec:methods-exp_phase}

Consider a population $I$ of $N$ objects (households, enterprises, aerial images, company websites or other text documents) and some target classification, or stratum, $s_i$ for each object $i \in I$. For now, we restrict ourselves to dichotomous categorical variables, i.e., $s_i \in \{0, 1\}$, where category~$1$ indicates the category of interest. A compelling example is the use of aerial images of rooftops to identify houses (the objects indexed by $i$) with solar panels ($s_i = 1$) \citep{curier2018monitoring}. From now on, we make three essential. Our first assumption is that there is some (possibly time consuming or otherwise expensive) way to retrieve the true category $s_i$ for each $i \in I$, for example by manually inspecting the aerial images and annotating them with a label indicating whether the image contains a solar panel. Our second assumption is that background variables or other features in the data contain sufficient information to estimate $s_i$ accurately. We draw a small random sample from the population and determine the true category $s_i$ for the objects in the sample. Then, the obtained data are, as usual, split at random into two sets. The first set is used to estimate model parameters (model selection and training). The second set, referred to as the test set $I_\text{test} \subset I$, is used to estimate the out-of-sample prediction error of the model. The number of observations in the test set is denoted by $n$ and we assume that $n \ll N$.

Consequently, the model can be used to produce an estimate $\widehat s_i$ of the true category to which object $i$ belongs. Here, our third assumption is that the success and misclassification probabilities of the model depend on $i$, but only through the true value of $s_i$. More precisely, we let $p_{ab}$ be the probability that $\widehat s_i = b$ given that $s_i = a$, for $a,b \in \{0, 1\}$. This specifies the \textit{classification error model} as introduced by \citet{bross_misclassification_1954}, following the notation in \citet{delden_accuracy_2016}. In addition, we adopt the notation $\bm{a}_i$, which is a 2-vector equal to $(1,0)$ if $s_i = 1$ and $(0,1)$ if $s_i=0$. The estimate $\bm{\widehat{a}}_i$ is defined similarly. The sum of all $\bm{a}_i$ is the 2-vector of counts $\bm{v}$. The first component of the 2-vector $\bm{\alpha} = \bm{v}/N$ is called the \textit{base rate} and is denoted by $\alpha$. It is immediate that $\E[\bm{\widehat \alpha}] = P^T\bm{\alpha}$, where $P$ is the confusion matrix with entries $p_{ab}$ (with $p_{11}$ as the top left entry). In general, $P^T\bm{\alpha} \neq \bm{\alpha}$, which indicates that $\widehat \alpha$ is a biased estimator for the base rate $\alpha$. The statistical bias of $\widehat \alpha$ as estimator for the base rate $\alpha$ is referred to as \textit{misclassification bias}.

A wide range of correction methods to reduce misclassification bias is available, see \citet{buonaccorsi_measurement_2010}. As briefly indicated in Section~\ref{sec:introduction}, \citet{kloos_comparing_2020} compared several correction methods aimed at improving the accuracy of estimators for $\alpha$. Two correction methods were most promising. The first correction method is the \textit{misclassification estimator} $\widehat \alpha_p$. It is defined as the first component of the following $2$-vector:
\begin{equation}
    \bm{\widehat \alpha}_p = \left(\widehat{P}^T\right)^{-1} \bm{\widehat{\alpha}},
\end{equation}
in which $\widehat{P}$ is the row-normalized confusion matrix obtained from the test set, i.e., with entries $\widehat{p}_{ab}~=~n_{ab}/n_{a+}$, where $n_{ab}$ denotes the number of objects $i$ in the test set for which $s_i~=~a$ and $\widehat{s}_i~=~b$ and where $n_{a+}$ denotes $n_{aa}+n_{ab}$. Moreover, the second correction method is the \textit{calibration estimator} $\widehat{\alpha}_c$. It is defined as the first component of the following $2$-vector:
\begin{equation}
    \bm{\widehat \alpha}_c = \widehat{C} \bm{\widehat{\alpha}},
\end{equation}
in which $\widehat{C}$ is the column-normalized confusion matrix obtained from the test set, i.e., with entries $\widehat{c}_{ab}~=~n_{ab} / n_{+b}$, where $n_{+b}$ denotes $n_{ab} + n_{bb}$. \citet{kloos_comparing_2020} have shown that if the test set is indeed a random sample from the target population, then the mean squared error of $\widehat{\alpha}_c$ is always smaller than that of $\widehat{\alpha}_p$.

\subsection{The production process of official statistics}\label{sec:methods-realistic}

Official statistics on a particular social or economic indicator are often produced for a certain period of time, at least annually, but often more frequently (quarterly or monthly). For as long as NSIs produce the official statistics on such an indicator, the output quality is required to be high. A challenging element in using classification algorithms in the production process of official statistics is that the target population $I$ changes over time, including the background variables $\bm{x}_i$ and the base rate $\alpha$. Therefore, the test set drawn at random from the population at one time period cannot be viewed as a random sample from the population at the next time period. A first solution would be to draw a new test set from the population (and then manually annotate the data) at each time period for as long as the statistical indicator is produced. However, due to cost constraints, such frequent data annotation is infeasible in practice. Thus, we will have to make an additional assumption to further investigate the results achieved by \citet{kloos_comparing_2020} in the context of a production process.

The additional assumption that we make is that the out-of-sample prediction accuracy of the model, i.e., the matrix $P$, is stable during a short period of time. More specifically, we assume (1) that $s_i$ causally determines the background variables $\bm{x}_i$ that are used in the model for $\widehat s_i$ and (2) that the causal relation does not change between (at least) two consecutive months or quarters. These two assumptions are identical to \textit{prior probability shift} as defined by \citet{moreno-torres_unifying_2012}. The first assumption, i.e., the causal relation between $s_i$ and $\bm{x}_i$, seems reasonable in many applications. In epidemiology, a disease causally determines the symptoms. In sentiment analysis, the writer's sentiment causally determines the words that the writer chooses. In land cover mapping, the mapped object causally determines the pixel values in the image. The second assumption (in terms of the classification error model) reads that $\Pr(\widehat{s}_i | s_i)$ does not change between consecutive months or quarters, but that $\alpha$ is allowed to change.

In the setting of prior probability shift, we consider two populations, namely the target population at two different moments in time, indicated by $I$ and $I'$, with sizes $N$ and $N'$. We assume that the test set $I_\text{test} \subset I$ of size $n$ has been obtained as a random sample from the target population $I$ in the first month or quarter, with true base rate $\alpha$. The aim is to estimate the base rate $\alpha'$ in the second month or quarter, i.e., within population $I'$, using prediction $\widehat s_i$ for $i \in I'$ and the estimates of $p_{ab}$ based on $I_\text{test} \subset I$. The type of concept drift that we investigate, prior probability shift, can be quantified by the difference $\delta \coloneqq \alpha' - \alpha$, which we will briefly refer to as the \textit{drift}. In the experimental phase we only consider a single population, which corresponds to putting $\delta = 0$. In Subsection \ref{sec:methods-theory}, we investigate the mean squared error of the calibration and misclassification estimator when $\delta \neq 0$.

\subsection{Theoretical results}\label{sec:methods-theory}
Expressions for the bias~$B$ and variance~$V$ of the misclassification estimator~$\alpha_p$ under drift~$\delta$ can be derived easily from the expressions presented by \citet{kloos_comparing_2020}. It follows that
\begin{align}
    & B[\hat \alpha_p]
    = \frac{1}{n(p_{00} + p_{11} - 1)^2} \cdot \left[ \frac{\alpha'}{\alpha}p_{11}(1-p_{11}) - \frac{1-\alpha'}{1-\alpha}p_{00}(1-p_{00}) \right] + O\left(\frac{1}{n^2}\right) \nonumber \\
    &\quad = \frac{p_{00} - p_{11}}{n(p_{00} + p_{11} - 1)} + \frac{\delta}{n(p_{00} + p_{11} - 1)^2} \cdot \left(\frac{p_{11}(1-p_{11})}{\alpha} + \frac{p_{00}(1-p_{00})}{1 - \alpha} \right) + O\left(\frac{1}{n^2}\right),
\end{align}
which is increasing in $\delta$ (but might first decrease in $\delta$ in absolute value). The variance of the misclassification estimator equals
\begin{equation}
    V(\hat \alpha_p) = \frac{(1-\alpha')^2 V(\hat p_{00}) + \alpha'^2 V(\hat p_{11})}{(p_{00} + p_{11} - 1)^2} + O\left(\frac{1}{n^2}\right),
\end{equation}
We neglect the terms of order $1/n^2$ and use Expressions~\eqref{eq:app_var_p11} and \eqref{eq:app_var_p00} from Appendix \ref{sec:appendix} to obtain
\begin{align}
    V(\hat \alpha_p)
    = \frac{1}{n(p_{00} + p_{11} - 1)^2} \cdot
      &\bigg[ T + 2\delta(p_{00} - p_{11})(p_{00} + p_{11}-1) \nonumber \\
      &\quad +\delta^2 \cdot \left(\frac{p_{11}(1-p_{11})}{\alpha} + \frac{p_{00}(1-p_{00})}{1 - \alpha} \right) \bigg] + O\left(\frac{1}{n^2}\right),
\end{align}
in which $T \coloneqq (1-\alpha)p_{00}(1-p_{00}) + \alpha p_{11}(1-p_{11})$. If $p_{00} \geq p_{11}$, then the variance increases as the drift $\delta$ increases. If $p_{00} < p_{11}$, then the effect of the drift is not immediately clear: increasing $\delta$ might decrease the variance, depending on the values of $\alpha$ and $\delta$. In Section~\ref{sec:results}, we will analyse the behaviour of $V(\hat \alpha_p)$ as function of $\alpha$ and $\delta$ numerically.

The expressions for the bias and variance of the calibration estimator presented by \citet{kloos_comparing_2020} were derived by conditioning on the base rate in the target population. If the drift $\delta$ is nonzero, that proof strategy breaks down. Therefore, we have adapted the proof to hold for nonzero $\delta$, resulting in the following expressions (see Expressions~\eqref{eq:bias_cal_drift} and \eqref{eq:var_cal_drift}).

\begin{theorem}\label{thm:bias_var_calibr_drift}
    The bias of $\hat \alpha_c$ as estimator for $\alpha$ under drift $\delta$ is given by
    \begin{equation}\label{eq:bias_cal_drift}
        B[\hat \alpha_c] = - \delta \frac{T}{\beta(1-\beta)} + O\left(\frac{1}{n^2}\right),
    \end{equation}
    in which $\beta \coloneqq (1-\alpha)(1-p_{00}) + \alpha p_{11}$ and $T = (1-\alpha)p_{00}(1-p_{00}) + \alpha p_{11}(1-p_{11})$. With that notation, the variance of $\hat \alpha_c$, under drift $\delta$, is given by
    \begin{align}\label{eq:var_cal_drift}
         V(\hat \alpha_c)
        &= \frac{\alpha(1-\alpha)}{n} \bigg[ \frac{T}{\beta(1-\beta)}  + 2\delta(p_{00} + p_{11} - 1) \left( \frac{p_{11}(1-p_{00})}{\beta^2} - \frac{p_{00}(1-p_{11})}{(1-\beta)^2} \right) \nonumber \\
        & \quad + \delta^2(p_{00} + p_{11} - 1)^2 \left( \frac{p_{11}(1-p_{00})}{\beta^3} + \frac{p_{00}(1-p_{11})}{(1-\beta)^3} \right) \bigg]
        + O\left(\frac{1}{n^2}\right).
    \end{align}
\end{theorem}

\begin{proof}
    See Appendix \ref{sec:appendix}.
\end{proof}

We make the following two observations: (1) the bias and the drift $\delta$ have opposite signs. and (2) the absolute bias is linearly increasing as a function of the absolute drift $|\delta|$. From these observations, the following sharp upper bound and lower bound for the absolute bias in terms of the absolute drift can be derived.

\begin{theorem}\label{thm:abs_bias_lower_bound}
    The absolute bias of $\hat \alpha_c$ as estimator for $\alpha' = \alpha + \delta$ is bounded from above by $|\delta|$.  If $p_{00} \leq p$ and $p_{11} \leq p$ for some $1/2 \leq p \leq 1$, then the absolute bias is at least $4p(1-p)|\delta|$.
\end{theorem}

\begin{proof}
    See Appendix \ref{sec:appendix}.
\end{proof}

The third observation is that, under prior probability shift, the bias of the misclassification estimator is still of order $1/n$ while that of the calibration estimator is nonzero if $\delta \neq 0$ and does not decrease for increasing $n$. This third observation is the key observation. The implication is that the conclusions drawn by \citet{kloos_comparing_2020} for the experimental phase of a machine learning project in official statistics do not hold when the algorithms are implemented in the production process. There, the drift $\delta$ is nonzero and a decision boundary arises. The aim of Section~\ref{sec:results} is to investigate the properties of the decision boundary.

%
%

\section{Results}\label{sec:results}

\begin{figure}
    \centering
    \includegraphics[width=0.7\textwidth]{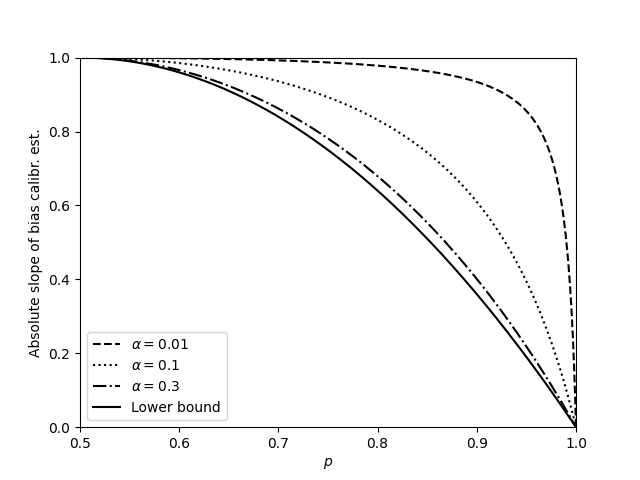}
    \caption{The slope of the bias of the calibration estimator $\hat \alpha_c$ as a function of the drift $\delta$ is equal to $-T/(\beta(1-\beta))$, which is strictly negative. The absolute value of that slope is plotted against the classification probability $p$, assuming that $p_{00} = p_{11} = p$, for four different values of $\alpha$. The solid black line depicts the theoretical lower bound (see Theorem~\ref{thm:abs_bias_lower_bound}) for the slope of the bias.}
    \label{fig:slope_abs_bias}
\end{figure}

The theoretical results from Section~\ref{sec:methods} indicate that in case $\delta$ is nonzero a decision boundary arises (between preferring (a) the misclassification estimator and (b) the calibration to reduce misclassification bias). The aim of this section is to understand that decision boundary. It is the main focus of Subsection~\ref{sec:results-boundary}. In advance, we investigate the bias under prior probability shift of the calibration estimator more closely in Subsection~\ref{sec:results-bias} and the difference in mean squared error between the two estimators in Subsection~\ref{sec:results-diff_mse}.

\subsection{Bias of the calibration estimator}\label{sec:results-bias}
We start plotting $T/(\beta(1-\beta))$, the absolute value of the slope of the bias of the calibration estimator, as a function of the classification probabilities for different values of $\alpha$, i.e., the base rate in the test set. For visualisation purposes, we restrict the function to $p_{00} = p_{11}$, parameterised by $p$. The results are depicted in Figure~\ref{fig:slope_abs_bias}, including the theoretical lower bound stated in Theorem~\ref{thm:abs_bias_lower_bound}. The slope of the bias as a function of $p$ is decreasing from $1$ at $p=0.5$ to $0$ at $p=1$. The smaller the value of $\alpha$, the later the function drops to 0. The reason is that the drift $\delta$ is defined as an absolute number and therefore it is relatively larger for smaller values of $\alpha$. From this observation we may conclude that the impact of (an absolute) drift $\delta$ on the bias of $\hat\alpha_c$ increases if $\alpha$ is further away from $0.5$, i.e., if the so-called \textit{class imbalance} increases.

\begin{figure}
    \centering
    \includegraphics[width = \textwidth]{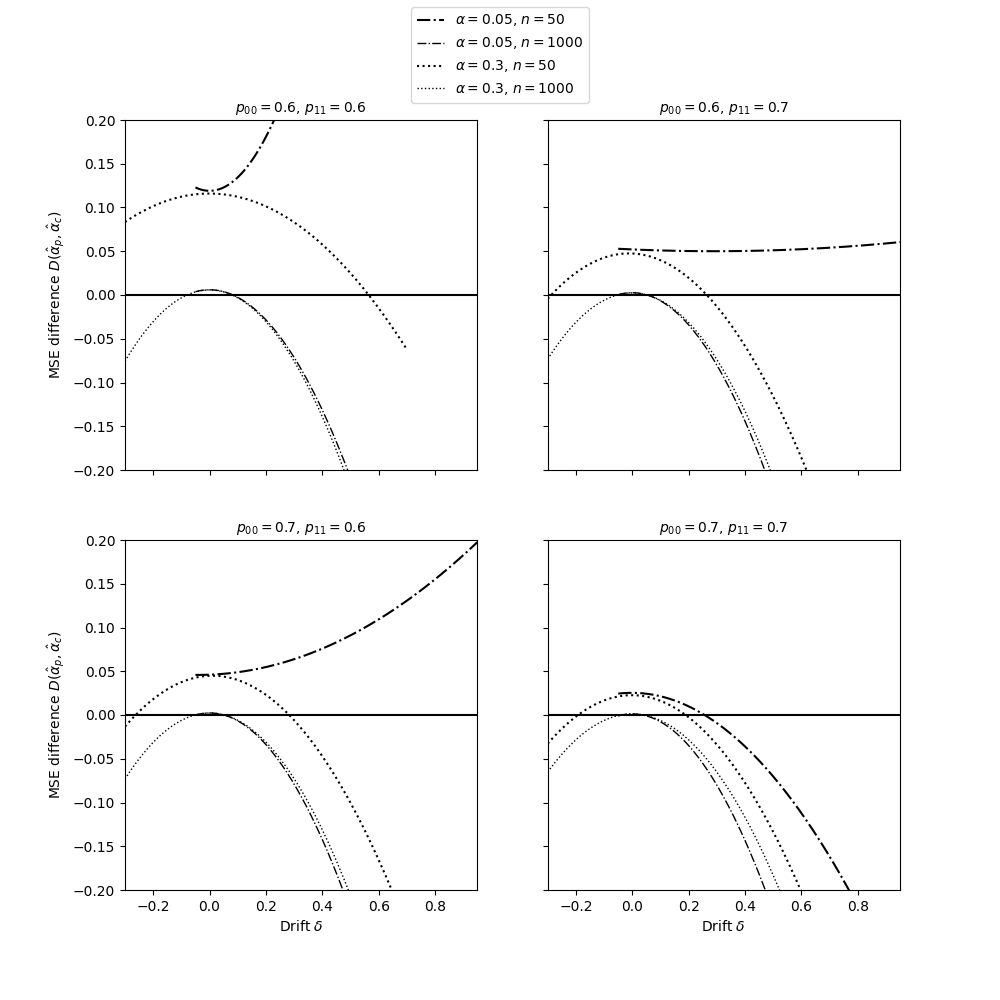}
    \caption{The difference $D(\hat\alpha_p, \hat\alpha_c)$ between the MSE of the misclassification estimator $\hat \alpha_p$ and that of the calibration estimator $\hat \alpha_c$, plotted as a function of $\delta$ for each possible combination of $\alpha \in \{0.05, 0.3\}$, $n \in \{50, 1000\}$ and $p_{00}, p_{11} \in \{0.6, 0.7\}$. Note that the drift $\delta$ ranges from $-\alpha$ to $1-\alpha$, because $\alpha' = \alpha + \delta$ must lie between 0 and 1.}
    \label{fig:mse_diff_function}
\end{figure}

\subsection{Difference in mean squared error}\label{sec:results-diff_mse}
Subsequently, we investigate the difference $D(\hat\alpha_p, \hat\alpha_c) \coloneqq MSE(\hat\alpha_p) - MSE(\hat\alpha_c)$ between the mean squared error of the misclassification estimator and that of the calibration estimator. The value of $D(\hat\alpha_p, \hat\alpha_c)$ as a function of $\delta$ is depicted in Figure~\ref{fig:mse_diff_function} for each possible combination of $\alpha \in \{0.05, 0.3\}$, $n \in \{50, 1000\}$ and $p_{00}, p_{11} \in \{0.6, 0.7\}$. Note that the drift $\delta$ ranges from $-\alpha$ to $1-\alpha$, because $\alpha' = \alpha + \delta$ must lie between 0 and 1. We report the following four observations. First, the difference is positive if $\delta = 0$ in any of the line plots, which corresponds to the main conclusion drawn by \citet{kloos_comparing_2020}. Second, when $n$ is sufficiently large (thin lines), the difference between the line plots are small. The reason is that the contribution of the variance terms is negligible compared to that of the squared bias of $\hat \alpha_c$, which does not depend on $n$ (see Theorem~\ref{thm:bias_var_calibr_drift}). Third, for highly imbalanced datasets combined with small test sets, i.e., $\alpha$ close to $0$ and $n$ small (thick dash-dotted lines), the variance of $\hat\alpha_p$ dominates if either $p_{00}$ is close to $0.5$ or $p_{11}$ is close to $0.5$. As a result, the calibration estimator has lowest mean squared error, independent of the magnitude of the drift $\delta$. Fourth, if the class distribution is relatively balanced (dotted lines), the difference $D(\hat\alpha_p, \hat\alpha_c)$ will become negative if $\delta$ increases, but the intersection moves farther away from $\delta = 0$ as $n$ decreases.

\subsection{The preferred estimator}\label{sec:results-boundary}
Finally, we compute, numerically, the unique positive value of $\delta$ (if it exists) at which the mean squared error of the misclassification and calibration estimator are identical. That is, we collect and reorganise the points of intersection $D(\hat\alpha_p, \hat\alpha_c)=0$ as discussed in Subsection~\ref{sec:results-diff_mse}. We view $D(\hat\alpha_p, \hat\alpha_c)$ as a map from $\R^3$ to $\R$ by fixing $\alpha$ and $n$ and using $\delta$, $p_{00}$ and $p_{11}$ as variables. Then, we plot the line within the two-dimensional surface $D(\hat\alpha_p, \hat\alpha_c) = 0$ where $p_{00} = p_{11}$, resulting in Figure~\ref{fig:min_delta}. Interestingly, the result is a decreasing function of $p$. At first, the result might seem to contradict the result obtained in the first analysis, cf. Figure~\ref{fig:slope_abs_bias}. There, the absolute slope of the bias as function of $\delta$ decreases with increasing $p$. Hence, the mean squared error of $\hat \alpha_c$ increases more slowly as a function of $\delta$ with increasing $p$. However, the result in Figure~\ref{fig:min_delta} follows from the fact that the difference in variance between $\hat \alpha_c$ and $\hat \alpha_p$ rapidly decreases as $p$ increases.

We stress that the lines in Figure~\ref{fig:min_delta} can be interpreted as decision boundaries. Each statistical indicator that is based on a classification algorithm plots somewhere in the $(p, \delta)$-plane depicted in Figure~\ref{fig:min_delta}. Our experimental result then reads as follows. If the plot of the indicator in the $(p, \delta)$-plane ends up above the decision boundary (which depends on $\alpha$ and $n$), then the misclassification estimator should be preferred over the calibration estimator to reduce misclassification bias. Otherwise, the calibration estimator should be preferred over the misclassification estimator. Moreover, in practice one should always compute the (estimated) bias and variance of the applied estimator, for they might still be high, e.g., when $n$ and $p$ are small and $\delta$ is large.

As a final remark, we indicate that these results hold if only the misclassification estimator and calibration estimator are considered. Admittedly, there may exist other estimators that might reduce misclassification bias even further. 

\begin{figure}
    \centering
    \includegraphics[width=0.7\textwidth]{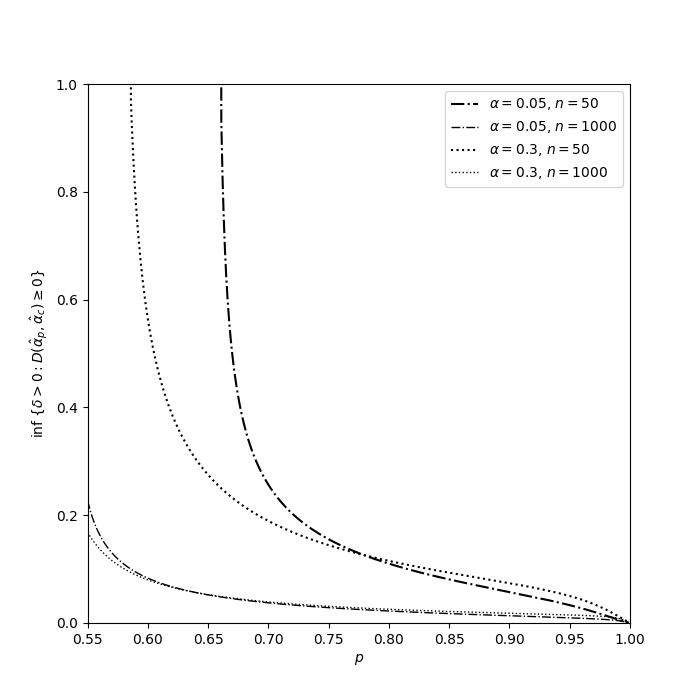}
    \caption{The unique positive value $\delta$ (if it exists) for which $D(\hat\alpha_p, \hat\alpha_c) = 0$, as a function of the classification probability $p$, assuming $p_{00} = p_{11} = p$. The lines should be interpreted as decision boundaries: below each of these lines the calibration estimator is preferred, while above each of the lines the misclassification estimator is preferred.}
    \label{fig:min_delta}
\end{figure}

%
%

\section{Conclusions and Discussion}\label{sec:discussion}

In this research, we investigated the output quality of official statistics based on classification algorithms. The main problem examined was how to reduce the bias caused by prior probability shift. We focused on two bias correction methods, namely (1) the misclassification estimator and (2) the calibration estimator. The results known for these two estimators failed to hold under prior probability shift. To obtain a further insight into the output quality of official statistics based on classification algorithms under prior probability shift, we adapted and extended the results achieved by \citet{kloos_comparing_2020} to hold for any value of the drift $\delta$. As theoretical results, we were able to show that (1) the calibration estimator is no longer unbiased and that (2) the absolute bias as a first-order approximation is a linearly increasing function of the absolute drift $|\delta|$ and does not depend on the test set size $n$.

Building on the theoretical results, we performed a simulation study consisting of three subsequent numerical analyses. The main conclusion drawn from the simulation results, is that the mean squared error of the calibration estimator is smaller than that of the misclassification estimator only when the performance of the classifier (in terms of $p_{00}$ and $p_{11}$) is low or when the drift $\delta$ is close to 0. The main conclusion has at least two significant implications. The first implication is that the conclusion gives a better understanding of the output quality of official statistics based on machine learning algorithms. More specifically, recommendations on which correction methods should be implemented in which situation are given. They allow for a more reliable implementation of machine learning algorithms in official statistics. The second implication is that the impact of the size and frequency of the training and test datasets is better understood. Essentially, our results show that the calibration estimator should \textit{not} be applied to data streams or time series data, unless training and test data in each time period are available to (a) retrain the classifier and hence (b) adapt to concept drift.

In case concept drift adaptation is considered too expensive due to cost constraints, the main conclusion (see above) implies that some minimal classification accuracy is required in order to use the misclassification estimator. To guarantee higher classification accuracy, more labelled training data have to be created, in general. In other words, NSIs should be careful when evaluating the cost efficiency of implementing machine learning algorithms for the production of official statistics. In the end, a substantial amount of high quality annotated data have to be created manually and consistently over a long period of time, which requires long-term investments in data analysts and domain experts.

Finally, we suggest three directions for future research. First, the robustness of classifier-based estimators should also be investigated for other types of concept drift, starting with the less restrictive type of prior probability shift as defined by \citet{webb_characterizing_2016}. Second, it might be worthwhile to examine methods for concept drift adaptation that are based on unlabelled data only, by carefully incorporating changes in the distribution of $P(X)$. Third, combinations or ensembles of different estimators require further research. We believe that a well-chosen combination of estimators will increase the overall robustness of classifier-based estimators under concept drift.

%
%

\newpage
\bibliography{concept_drift.bib}

%
%

\newpage
\appendix
\section{Appendix} \label{sec:appendix}
\input{appendix}

\end{document}

%% file: appendix.tex
This appendix contains the proofs of the theorems presented in the paper titled ``Improving the output quality of official statistics based on machine learning algorithms". For clarity, we will write $\widehat{\alpha}^*$ for the estimator based on the algorithms predictions $\widehat{s}_i$. In addition to the assumptions described in Section \ref{sec:methods}, we make two more technical assumptions, namely that $\widehat{\alpha}^*$ is independent of both the $\hat c_{ij}$ and the $\hat p_{ij}$. It follows that $\hat p_{00}$ and $\hat p_{11}$ are uncorrelated and that

\begin{equation}\label{eq:app_var_p11}
    V(\hat p_{11}) = \frac{p_{11}(1-p_{11})}{n \alpha} \left[1 + \frac{1 - \alpha}{n\alpha} \right]  + O\left(\frac{1}{n^3}\right).
\end{equation}
Similarly, the variance of $\hat p_{00}$ is given by
\begin{equation}\label{eq:app_var_p00}
    V(\hat p_{00}) = \frac{p_{00}(1-p_{00})}{n (1-\alpha)} \left[1 + \frac{\alpha}{n(1-\alpha)} \right] + O\left(\frac{1}{n^3}\right).
\end{equation}
For the proofs of these statements, consult Lemma 1 in the appendix of the paper by \citet{kloos_comparing_2020}. We will now provide the proof of Theorem \ref{thm:bias_var_calibr_drift} below.

\begin{proof}[Proof of Theorem \ref{thm:bias_var_calibr_drift}]
    
    Recall that the calibration estimator $\hat \alpha_c$ was given by
    \begin{equation}
        \hat\alpha_c = \hat \alpha^* \hat c_{11} + (1- \hat\alpha^*) \hat c_{10}.
    \end{equation}
    The derivations of the bias $B[\hat \alpha_c]$ and $V[\hat \alpha_c]$ are included below.
    
    \paragraph{Bias.}
    It is assumed that $\hat \alpha^*$ and $\hat c_{ij}$ are independent. Hence,
    \begin{equation}
        \E[\hat \alpha_c] = \E[\hat \alpha^*] \E[\hat c_{11}] + \E[1-\alpha^*] \E[\hat c_{10}].
    \end{equation}
    Recall the notation $\beta = (1-\alpha)(1-p_{00}) + \alpha p_{11}$ and set $\beta' \coloneqq (1-\alpha')(1-p_{00}) + \alpha' p_{11} = \E[\hat \alpha^*]$. To compute $\E[\hat c_{ij}]$, condition on $n_{1+}$, and note that $n_{0+} = n-n_{1+}$ is $n_{1+}$-measurable. It holds that  $\hat c_{11} \mid n_{1+} \overset{d}{=} X / (X+Y)$, with $X \sim Bin(n_{1+}, p_{11})$ and $Y \sim Bin(n_{0+}, 1-p_{00})$. Introducing the stochastic variable $\beta_+ \coloneqq n_{1+}p_{11} + n_{0+}(1-p_{00})$, a second-order Taylor approximation yields
    \begin{align}\label{eq:exp_c11}
        \E[\hat c_{11} \mid n_{1+}]
        &= \frac{n_{1+}p_{11}}{\beta_+}
        - \frac{n_{0+}(1-p_{00})}{\beta_+^3} n_{1+}p_{11}(1-p_{11})
        + \frac{n_{1+}p_{11}}{\beta_+^3} n_{0+}p_{00}(1-p_{00})
        + O\left( \frac{1}{n^2}\right) \nonumber \\
        &= \frac{n_{1+}p_{11}}{\beta_+} + p_{11}(1-p_{00})(p_{00}+p_{11}-1) \frac{n_{0+}n_{1+}}{\beta_+^3} + O\left( \frac{1}{n^2}\right). 
    \end{align}
    We then introduce the random variable $Z \sim Bin(n, \alpha)$ (i.e., $Z \overset{d}{=} n_{1+}$). Applying a Taylor approximation to the first term of Expression \eqref{eq:exp_c11} yields
    \begin{align}\label{eq:exp_c11_term1}
        \E\left[ \frac{n_{1+}p_{11}}{\beta_+} \right]
        &= \E\left[ \frac{p_{11}Z}{n(1-p_{00}) + (p_{00} + p_{11} - 1)Z} \right] \nonumber \\
        &= \frac{\alpha p_{11}}{\beta} - \frac{1}{2}\frac{2np_{11}(1-p_{00})(p_{00}+p_{11}-1)}{n^3\beta^3}n\alpha(1-\alpha) + O\left(\frac{1}{n^2}\right) \nonumber \\
        &= c_{11} - \frac{\alpha(1-\alpha)}{n}\frac{p_{11}(1-p_{00})(p_{00}+p_{11}-1)}{\beta^3} + O\left(\frac{1}{n^2}\right).
    \end{align}
    Next, apply a Taylor approximation to (the stochastic part of) the second term in \eqref{eq:exp_c11}:
    \begin{equation}\label{eq:exp_c11_term2}
        \E\left[ \frac{Z(n-Z)}{\beta_+^3} \right] = \frac{\alpha(1-\alpha)}{n\beta^3} + O\left(\frac{1}{n^2}\right).
    \end{equation}
    Combining \eqref{eq:exp_c11_term1} and \eqref{eq:exp_c11_term2} results in
    \begin{equation}\label{eq:exp_c11_final}
        \E[\hat c_{11}] = c_{11} + O\left(\frac{1}{n^2}\right) = \frac{\alpha p_{11}}{\beta} + O\left(\frac{1}{n^2}\right),
    \end{equation}
    where the second equality is included to stress that the result depends on $\alpha$, and not on $\alpha'$. Similarly, it follows that
    \begin{equation}\label{eq:exp_c10_final}
        \E[\hat c_{10}] = c_{10} + O\left(\frac{1}{n^2}\right) = \frac{\alpha(1-p_{11})}{1-\beta} + O\left(\frac{1}{n^2}\right).
    \end{equation}
    Substituting $\alpha' = \alpha + \delta$ and neglecting terms of order $1/n^2$ yields
    \begin{align}
        \E[\hat \alpha_c]
        &= \beta'\frac{\alpha p_{11}}{\beta} + (1-\beta')\frac{\alpha(1-p_{11})}{1-\beta} \nonumber \\
        &= \alpha p_{11} + \delta(p_{00} + p_{11} - 1) \frac{\alpha p_{11}}{\beta} + \alpha(1-p_{11}) + \delta(1 - p_{00} - p_{11}) \frac{\alpha(1-p_{11})}{1-\beta} \nonumber \\
        &= \alpha + \frac{\delta\alpha}{\beta(1-\beta)}\Big( (1-\beta)p_{11} - \beta(1-p_{11}) \Big)(p_{00} + p_{11} - 1) \nonumber \\
        & = \alpha + \frac{\delta\alpha(1-\alpha)(p_{00} + p_{11}-1)^2}{\beta(1-\beta)}.
    \end{align}
    It is straightforward to check that
    \begin{equation}
        \beta(1-\beta) - \alpha(1-\alpha)(p_{00} + p_{11}-1)^2 = \alpha p_{11}(1-p_{11}) + (1-\alpha)p_{00}(1-p_{00}) \eqqcolon T.
    \end{equation}
    Hence,
    \begin{equation}
        \E[\hat \alpha_c] = \alpha + \delta \left(\frac{\beta(1-\beta) - T}{\beta(1-\beta)} \right) + O\left(\frac{1}{n^2}\right) = \alpha' - \delta \frac{T}{\beta(1-\beta)} + O\left(\frac{1}{n^2}\right).
    \end{equation}
    Thus, we may conclude that the bias of $\hat \alpha_c$ as estimator of $\alpha'$ is equal to
    \begin{equation}
        B[\hat \alpha_c] = - \delta \frac{T}{\beta(1-\beta)} + O\left(\frac{1}{n^2}\right).
    \end{equation}

    \paragraph{Variance.}
    To compute the variance of $\hat \alpha_c$, we first note that
    \begin{equation}
        \E[(\hat \alpha^*)^2]
        = \E[\hat \alpha^*]^2 + V(\hat \alpha^*)
        = \E[\hat \alpha^*]^2 + O\left(\frac{1}{N}\right).
    \end{equation}
    A similar expression holds for the expectation of $(1-\hat \alpha^*)^2$ and that of $(1-\hat \alpha^*)\hat\alpha^*$. Neglecting the terms of order $1/N$, the above implies that
    \begin{align}
        V(\hat \alpha_c)
        &= V(\hat \alpha^* \hat c_{11}) + V((1- \hat \alpha^*)\hat c_{10}) + C(\hat \alpha^* \hat c_{11}, (1- \hat \alpha^*)\hat c_{10}) \nonumber \\
        &= \E[\hat \alpha^*]^2 V(\hat c_{11}) + \E[1-\hat \alpha^*]^2 V(\hat c_{10}) + \E[\hat \alpha^*]\E[(1-\hat \alpha^*)] C(\hat c_{11}, \hat c_{10}).
    \end{align}
    We may already substitute $\E[\hat \alpha^*] = \beta'$ in the above. It remains to derive expressions for $V(\hat c_{11})$, $V(\hat c_{10})$ and $C(\hat c_{11}, \hat c_{10})$. We compute $V(\hat c_{11})$ as $\E[\hat c_{11}^2] - \E[\hat c_{11}]^2$, because we have already derived an expression for the latter term. The random variable $\hat c_{11}^2 \mid n_{1+}$ is distributed as $X^2 / (X+Y)^2$. Setting $f(x,y) = x^2 / (x+y)^2$ yields
    \begin{equation}
        f_{xx}(x,y) = \frac{2y^2-4xy}{(x+y)^4}, \quad \text{and} \quad
        f_{yy}(x,y) = \frac{6x^2}{(x+y)^4}.
    \end{equation}
    It follows, neglecting terms of higher order, that
    \begin{align}\label{eq:var_c11}
        &\E[\hat c_{11}^2 \mid n_{1+}] \nonumber \\
        &= \frac{n_{1+}^2p_{11}^2}{\beta_+^2}
        + \frac{n_{0+}^2(1-p_{00})^2 - 2n_{1+}n_{0+} p_{11}(1-p_{00})}{\beta_+^4} n_{1+}p_{11}(1-p_{11}) + \frac{3n_{1+}^2p_{11}^2}{\beta_+^4}n_{0+}p_{00}(1-p_{00}) \nonumber \\
        &= \frac{n_{1+}^2p_{11}^2}{\beta_+^2}
        + p_{11}(1-p_{00})\frac{n_{1+}n_{0+}
        \Big( n(1-p_{00})(1-p_{11}) + n_{1+}(p_{00}+p_{11}-1)(2p_{11}+1)\Big)}{\beta_+^4}.
    \end{align}
    Again, let $Z \sim Bin(n, \alpha)$ and consider the function $f(z) = z^2/(A+Bz)^2$, with $A = n(1-p_{00})$ and $B = (p_{00} + p_{11} - 1)$. Then
    \begin{equation}
        f_{zz}(z) = \frac{2A^2 - 4ABz}{(A+Bz)^4}.
    \end{equation}
    The conditional expectation then equals (up to terms of order $1/n^2$):
    \begin{align}\label{eq:var_c11_term1}
        \E\left[\frac{n_{1+}^2p_{11}^2}{\beta_+^2}\right]
        &= \E\left[\frac{p_{11}^2Z^2}{(A+BZ)^2}\right] \nonumber \\
        &= \frac{\alpha^2 p_{11}^2}{\beta^2} + p_{11}^2\frac{n^2(1-p_{00})^2 - 2n^2\alpha(1-p_{00})(p_{00} + p_{11} - 1)}{n^4\beta^4}n\alpha(1-\alpha) + O\left(\frac{1}{n^2}\right) \nonumber \\
        &= c_{11}^2 + \frac{\alpha(1-\alpha)}{n} \frac{p_{11}^2(1-p_{00}) \Big(1-p_{00} - 2\alpha(p_{00}+p_{11}-1) \Big)}{\beta^4} + O\left(\frac{1}{n^2}\right).
    \end{align}
    Apply a Taylor approximation to (the stochastic part of) the second term in \eqref{eq:var_c11} to obtain:
    \begin{align}\label{eq:var_c11_term2}
        \frac{\alpha(1-\alpha)}{n}\frac{p_{11}(1-p_{00})\Big( (1-p_{00})(1-p_{11}) + \alpha(p_{00}+p_{11}-1)(2p_{11}+1)\Big)}{\beta^4} + O\left(\frac{1}{n^2}\right).
    \end{align}
    At last, combining \eqref{eq:var_c11_term1} and \eqref{eq:var_c11_term2}, and subtracting \eqref{eq:exp_c11_final} squared, the variance of $\hat c_{11}$ can be expressed as
    \begin{equation}\label{eq:var_c11_final}
        V(\hat c_{11}) = \frac{\alpha(1-\alpha)}{n}\frac{p_{11}(1-p_{00})}{\beta^3} + O\left(\frac{1}{n^2}\right).
    \end{equation}
    Similarly, it can be shown that
    \begin{equation}\label{eq:var_c10_final}
        V(\hat c_{10}) = \frac{\alpha(1-\alpha)}{n}\frac{p_{00}(1-p_{11})}{(1-\beta)^3} + O\left(\frac{1}{n^2}\right).
    \end{equation}
    Moreover, it can be shown that $\hat c_{11}$ and $\hat c_{10}$ are uncorrelated, using the same strategy that was used to prove that $\hat p_{00}$ and $\hat p_{11}$ are uncorrelated. For completeness:
    \begin{align}
        \E[\hat c_{11} \hat c_{10}]
        &= \E\left[ \E\left[\left.\frac{n_{11}n_{10}}{n_{+1}n_{+0}} \right| n_{+1}\right] \right] \nonumber \\
        &= \E\left[ \frac{1}{n_{1+}n_{0+}} \E\left[\left.n_{11}n_{10} \right| n_{+1}\right] \right] \nonumber \\
        &= \E\left[ \frac{1}{n_{1+}n_{0+}} \cdot n_{+1}c_{11}n_{+0}c_{10}\right]
        = c_{11}c_{10} = \E[\hat{c}_{11}] \E[\hat{c}_{10}].
    \end{align}
    It implies that $C(\hat c_{11}, \hat c_{10}) = \E[\hat c_{11} \hat c_{10}] - \E[\hat c_{11}] \E[\hat c_{10}] = 0$. Finally, we may conclude that
    \begin{equation}
        V(\hat \alpha_c)
        = \frac{\alpha(1-\alpha)}{n}\left( \beta'^2 \frac{p_{11}(1-p_{00})}{\beta^3} + (1-\beta')^2 \frac{p_{00}(1-p_{11})}{(1-\beta)^3} \right) + O\left(\frac{1}{n^2}\right).
    \end{equation}
    Substituting $\alpha' = \alpha + \delta$ yields
    \begin{align}
        V(\hat \alpha_c)
        &= \frac{\alpha(1-\alpha)}{n} \bigg[ \frac{T}{\beta(1-\beta)}  + 2\delta(p_{00} + p_{11} - 1) \left( \frac{p_{11}(1-p_{00})}{\beta^2} - \frac{p_{00}(1-p_{11})}{(1-\beta)^2} \right) \nonumber \\
        & \quad + \delta^2(p_{00} + p_{11} - 1)^2 \left( \frac{p_{11}(1-p_{00})}{\beta^3} + \frac{p_{00}(1-p_{11})}{(1-\beta)^3} \right) \bigg]
        + O\left(\frac{1}{n^2}\right).
    \end{align}
    The expression above completes the derivation of the variance of the calibration estimator under prior probability shift.
\end{proof}

To prove the Theorem \ref{thm:abs_bias_lower_bound}, we need the following lemma.

\begin{lemma}\label{lemma:partial_derivatives}
    The slope of the absolute value of the first-order approximation of the bias of the calibration estimator as a function of the absolute value $|\delta|$ of the prior probability shift is decreasing in $p_{00}$ and $p_{11}$ for all $1/2 \leq p_{00} \leq 1$ and $1/2 \leq p_{11} \leq 1$.
\end{lemma}

\begin{proof}
    We introduce the notation $x = p_{00}$, $y = p_{11}$ and $\beta = \beta(x,y, \alpha) = (1-\alpha)(1-x) + \alpha y$. We then define the functions
    \begin{equation}
    f(x,y, \alpha) = \frac{(1-x)y}{\beta} \quad \text{and} \quad g(x,y, \alpha) = \frac{x(1-y)}{1-\beta}.
    \end{equation}
    The function $h = f+g$ then satisfies $ |\delta| \cdot h(p_{00}, p_{11}, \alpha) = \big|B[\hat\alpha_c]\big|$ up to terms of order $1/n^2$. We will examine the sign of the partial derivatives of $h$ with respect to $x$ and $y$, which we denote by $h_x$ and $h_y$, respectively. To that end, we first compute the partial derivatives of $f$ and $g$, giving
    \begin{equation}
    f_x(x,y, \alpha) = \frac{-\alpha y^2}{\beta^2} \quad \text{and} \quad
    g_x(x,y, \alpha) = \frac{\alpha(1-y)^2}{(1-\beta)^2}.
    \end{equation}
    Hence,
    \begin{equation}\label{eq:part_deriv_h_x}
        h_x(x,y, \alpha) = \frac{\alpha}{\beta^2(1-\beta)^2} \cdot \Big( ((1-y)\beta)^2 - (y(1-\beta))^2 \Big).
    \end{equation}
    Setting this to zero yields $(1-y)\beta = y(1-\beta)$ or $(1-y)\beta = -y(1-\beta)$. As $ 1/2 \leq x,y \leq 1$ and $0 < \alpha < 1$ it follows that $ 1-x \leq \beta \leq y$ with equality if and only if $1-x = y$, i.e. $x=y = 1/2$. It implies that $(1-y)\beta$ is nonnegative and that $y(1-\beta)$ is strictly positive, hence the equation $(1-y)\beta = -y(1-\beta)$ has no solution. Moreover, it implies that $(1-y)\beta \leq y(1-\beta)$ with equality only at $x=y=1/2$. From this we may conclude that $h$ is decreasing in $x$ for all $1/2 < x \leq 1$ and that $h_x(\tfrac{1}{2}, \; \cdot \;, \; \cdot \;) = 0$.
    
    The partial derivatives $h_x$ and $h_y$ can be related through a simple symmetry argument: it holds that $\beta(y,x,\alpha) = 1-\beta(x, y, 1-\alpha)$, which implies that $h(y,x,\alpha) = h(x,y,1-\alpha)$. Consequently, it holds that $h_y(\; \cdot \;, \; \cdot \;, \alpha) = h_x(\; \cdot \;, \; \cdot \;, 1-\alpha)$. It follows that $h$ is also decreasing in $y$ for all $1/2 < y \leq 1$ and that $h_y(\; \cdot \;, \tfrac{1}{2}, \; \cdot \;) = 0$.
    
    We conclude that the slope $h$ of the first-order approximation of the bias of the calibration estimator under prior probability shift is decreasing in $p_{00}$ and $p_{11}$ for $1/2 \leq p_{00}, p_{11} \leq 1$, attaining its global maximum at $p_{00} = p_{11} = 1/2$, where $h = 1$ and $\big|B[\hat\alpha_c]\big| = |\delta|$.
\end{proof}

The statement of Theorem \ref{thm:abs_bias_lower_bound} is an immediate consequence of the lemma above.

\begin{proof}[Proof of Theorem \ref{thm:abs_bias_lower_bound}]
    Lemma \ref{lemma:partial_derivatives} implies that $|B[\hat\alpha_c]| \leq |\delta|$ and that $|B[\hat \alpha_c]| \geq |\delta| \cdot h(p,p, \alpha)$. To simplify the latter, observe that $T(p,p,\alpha) = p(1-p)$ and that $0 \leq 1-p < \beta(p,p,\alpha) < p \leq 1$, using that $1/2 \leq p \leq 1$ and $0 < \alpha < 1$. It follows that $\beta(1-\beta) \leq 1/4$, which completes the proof.
\end{proof}